\documentclass[runningheads,a4paper]{llncs}

\usepackage{amssymb}
\usepackage{graphicx}

\usepackage{url}
\urldef{\mailsa}\path|{alfred.hofmann, ursula.barth, ingrid.haas, frank.holzwarth,|
\urldef{\mailsb}\path|anna.kramer, leonie.kunz, christine.reiss, nicole.sator,|
\urldef{\mailsc}\path|erika.siebert-cole, peter.strasser, lncs}@springer.com|
\newcommand{\keywords}[1]{\par\addvspace\baselineskip
\noindent\keywordname\enspace\ignorespaces#1}

\newtheorem{statement}[theorem]{Statement}

\begin{document}

\mainmatter

\title{Stabbing line segments with disks: complexity and approximation algorithms
\footnote{This work was supported by Russian Science Foundation, project №14-11-00109.}}

\titlerunning{Stabbing line segments with disks: complexity and approximation algorithms}

\author{Konstantin Kobylkin}

\authorrunning{Konstantin Kobylkin}

\institute{Institute of Mathematics and Mechanics, Ural Branch of RAS,\\
Sophya Kovalevskaya str. 16, 620990 Ekaterinburg, Russia,\\
Ural Federal University, Mira str. 19, 620002 Ekaterinburg, Russia,\\
\url{kobylkinks@gmail.com}}

%
%

\toctitle{Stabbing line segments with disks: complexity and approximation algorithms}
\tocauthor{Konstantin Kobylkin}
\maketitle

\begin{abstract}
Computational complexity and approximation algorithms are reported
for a problem of stabbing a set of straight line segments
with the least cardinality set of disks of fixed radii $r>0$ where the set of segments
forms a straight line drawing $G=(V,E)$ of a planar graph without edge crossings.
Close geometric problems arise in network security applications.
We give strong NP-hardness of
the problem for edge sets of Delaunay triangulations, Gabriel graphs and other subgraphs (which are often used
in network design) for $r\in [d_{\min},\eta d_{\max}]$ and some constant $\eta$
where $d_{\max}$ and $d_{\min}$ are Euclidean lengths of the longest and shortest graph edges respectively.
Fast $O(|E|\log|E|)$-time $O(1)$-approximation algorithm is
proposed within the class of straight line drawings of planar graphs
for which the inequality $r\geq \eta d_{\max}$ holds uniformly for
some constant $\eta>0,$ i.e. when lengths of edges of $G$ are uniformly bounded from
above by some linear function of $r.$
\end{abstract}

\section{Introduction}
Numerous applications from security, sensor placement and robotics lead to
computational geometry problems in which one needs to find the smallest
cardinality set $C$ of points on the plane having bounded (in some sense) visibility area such
that each piece of the boundary of a given geometric object or any part of the complex (i.e.
set of edges or faces) of a plane graph is within visibility area of some point from $C,$
see e.g. \cite{journals/comgeo/BoseKL03}, \cite{orourke_art_1987}.
Refining complexity
statuses and designing approximation algorithms for these problems is still an area
of active research. In this paper complexity and approximability are studied of the following problem.

\noindent {\sc Intersecting Plane Graph with Disks (IPGD)}:
given a straight line drawing (or a plane graph) $G=(V,E)$
of an arbitrary simple\footnote{a graph without loops and
parallel edges} planar graph without edge crossings
and a constant $r>0,$ find the smallest cardinality
set $C\subset\mathbb{R}^2$ of points (disk centers)
such that each edge $e\in E$ is within Euclidean distance $r$
from some point $c=c(e)\in C$ or, equivalently, the disk of radius $r$ centered at $c$
intersects $e.$

The {\sc IPGD} abbreviation is used throughout our paper to denote
the above problem for simplicity of presentation.

Applications of complexity and algorithmic
analysis of the {\sc IPGD} problem come from network security. More specifically, {\sc IPGD} represents the
following model in which we are to evaluate vulnerability of
some physical network to simultaneous technical failures
caused by natural (e.g. floods, fire, electromagnetic pulses) and human sources.
In this model network nodes are modeled
by points on the plane while its physical links are given in the form of straight line segments.
A catastrophic event (threat) is usually localized in a particular
geographical area and modeled by a disk of some fixed radius $r>0.$
A threat impacts a network link when the corresponding disk and segment intersect.
Evaluation of the network vulnerability can be posed in the form of finding the minimum number of
threats along with their positions that cause all network links to be broken.
Thus, it brings us to the {\sc IPGD} problem assuming that
network links are geographically non-overlapping. A similar
setting is considered in \cite{journals/ton/AgarwalEGHSZ13}
with the fixed number of threats. Furthermore, in \cite{orourke_art_1987} a close geometric problem is considered
called the Art Gallery problem where point coverage area is affected by boundaries of
its neighbouring geometric objects whereas point has circular visibility area in the
case of the {\sc IPGD} problem.

In this paper computational complexity and approximability of {\sc IPGD} are studied for simple plane graphs with
either $r\in [d_{\min},d_{\max}]$ or $r=\Omega(d_{\max})$ where $d_{\max}$ and $d_{\min}$ are
Euclidean lengths of the longest and shortest edges of $G.$ Our emphasis is on those classes of simple plane graphs that are
defined by some distance function, namely, on Delaunay triangulations, some of
their connected subgraphs, e.g. for Gabriel graphs.
These graphs are often called {\it proximity} graphs.
Delaunay triangulations are plane graphs which admit efficient geometric
routing algorithms \cite{conf/swat/BoseCDT14}, thus, representing convenient network topologies. Gabriel
graphs arise in modeling
wireless networks \cite{journals/winet/BoseMSU01}.

\subsection{Related work}

{\sc IPGD} is related to several well-known combinatorial optimization
problems. First, we have the Continuous Disk Cover ({\sc CDC}) problem
for the case of {\sc IPGD} where $G$ consists of isolated vertices,
i.e. when segments from $E$ are all of zero length.
Strong NP-hardness is well known for {\sc CDC} \cite{MasuyamaIbarakiEtAl1981}. Second,
{\sc IPGD} coincides with known {\sc Vertex Cover} problem for $r=0.$
Third, it is the special case of the geometric {\sc Hitting Set} problem on the plane.

\noindent {\sc Hitting Set}:
given a family ${\cal{N}}$ of sets on the plane and a set $U\subseteq\mathbb{R}^2,$
find the smallest cardinality set $H\subseteq U$
such that $N\cap H\neq\varnothing$ for every $N\in {\cal{N}}.$

{\sc IPGD} coincides with {\sc Hitting Set} if we set ${\cal{N}}:={\cal{N}}_r(E)=\{N_r(e)\}_{e\in E}$
and $U:=\mathbb{R}^2$ where $N_r(e)=B_r(0)+e=\{x+y:x\in B_r(0),\,y\in e\}$ is Euclidean $r$-neighbourhood of $e$ having form of
Minkowski sum and $B_r(x)$ is the disk of radius $r$ centered at $x\in\mathbb{R}^2.$
An {\it aspect ratio} of a closed convex set $N$ with ${\mathrm{int}}\,N\neq\varnothing$\footnote{${\mathrm{int}}\,N$ is the set
of interior points of $N$}
coincides with the ratio of the minimum radius of the disk which contains $N$ to the
maximum radius of the disk which is contained in $N.$
For example, each set $N_r(e)$ (also called by object in the sequel) has aspect ratio equal to $1+\frac{d(e)}{2r}$ where $d(e)$
is Euclidean length of edge $e\in E.$
APX-hardness of the discrete\footnote{when $U$ coincides with some prescribed finite set}
{\sc Hitting Set} problem is presented for families of axis-parallel rectangles, generally, with unbounded
aspect ratio, \cite{journals/comgeo/ChanG14},
and for families of triangles of bounded aspect ratio \cite{conf/esa/Har-PeledQ15}.

\subsection{Results}

Our results report complexity
and approximation algorithms for the {\sc IPGD} problem within
several classes of plane graphs under different assumptions on $r.$
Let $S$ be a set of $n$ points in general position on the plane no four of which are cocircular.
We call a plane graph $G=(S,E)$ a {\it Delaunay
triangulation} if $[u,v]\in E$
iff there is a disk $T$ such that
$u,v\in {\rm{bd}}\,T$\footnote{${\rm{bd}}\,T$ denotes the set of boundary points of $T$} and $S\cap {\rm{int}}\,T=\varnothing.$
Finally, a plane graph $G=(S,E)$ is named a {\it nearest neighbour} graph when
$[u,v]\in E$ iff either $u$ or $v$ is the nearest Euclidean neighbour for $v$ or $u$ respectively.

\subsubsection{Hardness results.} Our first result claims strong NP-hardness of {\sc IPGD}
within the class of Delaunay triangulations and some known classes of their connected subgraphs
(Gabriel and relative neighbourhood graphs) for $r\in [d_{\min},d_{\max}]$ and $\mu=\frac{d_{\max}}{d_{\min}}=O(|S|).$
{\sc IPGD} remains strongly NP-hard within the class of nearest neighbour graphs
for $r\in [d_{\max},\eta d_{\max}]$ with a large constant $\eta$ and $\mu\leq 4.$
Furthermore, we have the same NP-hardness results under the same restrictions on $r$ and $\mu$
even if we are bound to choose points of $C$ close to vertices of $G.$
The upper bound on $\mu$  for Delaunay triangulations is comparable with the lower bound $\mu=\Omega\left(\sqrt[3]{n^2}\right)$
which holds true with positive probability for Delaunay triangulations produced by
$n$ random independent points on the unit disk \cite{journals/comgeo/ArkinAMM15}.
Thus, declared restrictions on $r$ and $\mu$ define natural instances of {\sc IPGD}.

An upper bound on $\mu$ implies an upper bound on the
ratio of the largest and smallest aspect ratio of objects from ${\cal{N}}_r(E).$
The {\sc Hitting Set} problem is generally easier when sets from ${\cal{N}}$ have almost equal aspect ratio
bounded from above by some constant.
Our result for the class of nearest neighbour graphs gives the problem NP-hardness
in the case where objects of ${\cal{N}}_r(E)$ have almost equal constant aspect ratio.

In distinction to known results for the {\sc Hitting Set} problem mentioned above our study is mostly for
its continuous setting with the structured system ${\cal{N}}_r(E)$ formed
by an edge set of a specific plane graph; each set from ${\cal{N}}_r(E)$ is of the
special form of Minkowski sum of some graph edge and radius $r$ disk. Our proofs are
elaborate complexity reductions from the {\sc CDC} problem which is intimately related
to {\sc IPGD}.

\subsubsection{Positive results.} Let $R(E)$ be the smallest radius of the disk that
intersects all segments from the edge set $E.$
As opposed to the cases where either $r\in [d_{\min},d_{\max}]$ or $r\in [d_{\max},\eta d_{\max}],$
{\sc IPGD} is solvable within the class of simple plane graphs, for which
the inequality $r\geq\eta R(E)$ holds uniformly for some fixed $\eta>0,$
in $O\left(k^2|E|^{2k+1}\right)$ time with $k=\left\lceil\frac{\sqrt{2}}{\eta}\right\rceil^2.$
Above inequality implies an upper bound $k$ on its optimum.
Taking proof of $W[1]$-hardness into account of parameterized version of {\sc CDC} \cite{conf/esa/Marx05} as well as
the reduction used to prove the theorem \ref{kobylkinks123454321} of this paper,
it seems unlikely to improve this
time bound to $O(f(k)|E|^c)$ for any computable function $f$ and any constant $c>0.$

Finally, we present an $8p(1+2\lambda)$-approximation $O(|E|\log|E|)$-time algorithm
for {\sc IPGD} when the inequality $r\geq \frac{d_{\max}}{2\lambda}$ holds true uniformly within a class of
simple plane graphs for a constant $\lambda>0,$ where $p(x)$ is
the smallest number of unit disks needed to cover any disk of radius $x>1.$
It corresponds to the case where segments from $E$ have their lengths uniformly bounded from
above by some linear function of $r,$ or, in other words,
when objects from ${\cal{N}}_r(E)$ have their aspect ratio bounded from above by $1+\lambda.$
A similar but more complex $O(|E|^{1+\varepsilon})$-time constant
factor approximation algorithm is given in \cite{efrat} to approximate the
{\sc Hitting Set} problem for sets of objects whose
aspect ratio is bounded from above by some constant.

\section{NP-hardness results}

We give complexity analysis for the {\sc IPGD} problem by considering
its setting where $r\in [d_{\min},d_{\max}].$ Under this restriction on $r$
{\sc IPGD} coincides neither with known {\sc Vertex Cover} problem nor with {\sc CDC}.
In fact it is equivalent (see the Introduction) to the geometric {\sc Hitting Set} problem for
the set ${\cal{N}}_r(E)$ of Euclidean $r$-neighbourhoods of edges of $G.$
For the {\sc IPGD} problem we claim its NP-hardness even
if we restrict the graph $G$ to be either a Delaunay triangulation or some of its known subgraphs. We keep the ratio
$\mu=\frac{d_{\max}}{d_{\min}}$ bounded from above, thus, imposing an
upper bound on the ratio of the largest and smallest aspect ratio of objects from ${\cal{N}}_r(E).$
We show that {\sc IPGD} remains intractable even in its simple case where $r=\Theta(d_{\max})$ and
$\mu$ is bounded by some small constant
or, equivalently, when objects of ${\cal{N}}_r(E)$ have close constant aspect ratio.

Our first hardness result for {\sc IPGD} is obtained by using a complexity reduction from the {\sc CDC} problem.
Below we describe a class of hard instances of the {\sc CDC} problem which correspond to hard instances
of the {\sc IPGD} problem for Delaunay triangulations with relatively small
upper bound on the parameter $\mu.$

\subsection{NP-hardness of the {\sc CDC} problem}

To single out the
class of hard instances of the {\sc CDC} problem a reduction is used in \cite{MasuyamaIbarakiEtAl1981} from
the strongly NP-complete minimum dominating set problem which is formulated as follows: given
a simple planar graph $G_0=(V_0,E_0)$
of degree at most $3,$
find the smallest cardinality set $V'_0\subseteq V_0$
such that for each $u\in V_0\backslash V'_0$ there is some
$v=v(u)\in V'_0$ which is adjacent to $u.$

Below an integer grid denotes the set of all points on the plane with
integer-valued coordinates each of which belongs to some bounded interval.
An {\it orthogonal} drawing of the graph $G_0$ on some integer grid is the drawing whose vertices
are represented by points on that grid whereas its edges are given in the form of polylines
that are composed of connected axis-parallel straight line segments of the form $[p_1,p_2],\,[p_2,p_3],
\ldots,\,[p_{k-1},p_k],$ and intersecting only at the edge endpoints $p_1$ and $p_k,$ where each point $p_i$ again belongs
to the grid. In \cite{MasuyamaIbarakiEtAl1981} strong NP-hardness of {\sc CDC} is proved
by reduction from the minimum dominating set problem.
This reduction involves using plane orthogonal drawing of $G_0$ on some
integer grid. More specifically, a set $D$ is build on that
grid with $V_0\subset D.$
The resulting hard instance of the {\sc CDC} problem is for the set $D$ and some integer (constant) radius $r_0\geq 1.$
Let us observe that $G_0$ admits an orthogonal drawing (theorem 1 \cite{tamassia1989}) on the grid
of size $O(|V_0|)\times O(|V_0|)$ whereas total length of each edge is of the order $O(|V_0|).$
Proof of strong NP-hardness of {\sc CDC} could
be conducted taking into account this observation. We can formulate (see theorems 1 and 3 from \cite{MasuyamaIbarakiEtAl1981})

\begin{theorem}{\rm{\cite{MasuyamaIbarakiEtAl1981}}}\label{kobylkinTh00}
The {\sc CDC} problem is strongly NP-hard for a constant integer radius $r_0$
and point sets $D$ on the integer grid
of size $O(|D|)\times O(|D|).$ It remains strongly NP-hard even if we restrict centers of
radius $r_0$ disks to be at the points of $D.$
\end{theorem}

\begin{remark}\label{kobylkin474757428}
For every simple planar graph $G_0$ of degree at most $3$ its orthogonal drawing can be
constructed such that at least one its edges is a polyline which is composed of at least
two axis-parallel segments.
\end{remark}

\subsection{NP-hardness of the {\sc IPGD} problem for Delaunay triangulations}
To build a reduction from the {\sc CDC} problem on the set $D$ (as constructed in proof
of the theorem \ref{kobylkinTh00} from \cite{MasuyamaIbarakiEtAl1981}),
we exploit a simple idea that a radius $r$ disk covers a set of points $D'\subset D$
iff a slightly larger disk intersects (and, sometimes, covers) straight line segments, each of which
is close to some point of $D'$ and has a small length with respect
to distances between points of $D.$ Then a proximity graph $H$ is build whose vertex set
coincides with the set of endpoints of small segments corresponding to points of $D.$
Since $H$ usually contains these small segments as its edges, this technique gives NP-hardness for the {\sc IPGD}
problem within numerous classes of proximity graphs.
The following technical lemma holds which reports an $r$-dependent
lower bound on the distance between any point with integer
coordinates and a radius $r$ circle through the pair of integer-valued
points.

\begin{lemma}\label{kobylkinks1234321}
Let $X\subset\mathbb{Z}^2,$ $r\geq 1$ is an integer,
$\rho(u;v,w)$ denotes the minimum of two Euclidean distances
from an arbitrary point $u\in X$ to the union of two radius $r$ circles which pass through
distinct points $v$ and $w$ from $X,$ where $|v-w|_2\leq 2r,$
$\mathbb{Z}$ is the set of integers and $|\cdot|_2$ is Euclidean norm.
Then
$$
\min\limits_{u\notin C(v,w),\, v\neq w,\, u,v,w\in X,\, |v-w|_2\leq 2r}
\rho(u;v,w)\geq\frac{1}{480r^5},
$$
where $C(v,w)$ is the union of two radius $r$ circles passing through $v$ and $w.$
\end{lemma}

Let us formulate the following restricted form of {\sc IPGD}.

\noindent {\sc Vertex Restricted IPGD (VRIPGD($\delta$)):}
given a simple plane graph $G=(V,E),$ a constant $\delta>0$
and a constant $r>0,$ find the least cardinality
set $C\subset\mathbb{R}^2$
such that each $e\in E$ is within Euclidean distance $r$
from some point $c=c(e)\in C$ and $C\subset\bigcup\limits_{v\in V}B_{\delta}(v).$

\begin{theorem}\label{kobylkinks123454321}
Both {\sc IPGD} and {\sc VRIPGD}$(\delta)$ problems are strongly NP-hard for
$r\in [d_{\min}, d_{\max}],\,\mu=O(n)$ and $\delta=\Theta(r)$
within the class of Delaunay triangulations, where $n$ is the number of vertices in triangulation.
\end{theorem}
\begin{proof}
Let us prove that {\sc IPGD} is strongly NP-hard. Proof technique for the
{\sc VRIPGD}($\delta$) problem is analogous taking into
account the theorem \ref{kobylkinTh00} (see also proof of the theorem 3 from \cite{MasuyamaIbarakiEtAl1981}
for details).
For any hard instance of the {\sc CDC} problem, which the theorem \ref{kobylkinTh00} reports,
the {\sc IPGD} problem instance is built for $r=r_0+\delta$ and $\delta=\frac{1}{2000^22r_0^{11}}$ as follows.
For every $u\in D$ points $u_0$ and $v_0$ are found such that $|u-u_0|_{\infty}\leq\delta/2$ and $|u-v_0|_{\infty}\leq\delta/2,$
where $I_u=[u_0,v_0]$ has Euclidean length at least $\delta/2$ and $|\cdot|_{\infty}$ denotes norm
in $\mathbb{R}^2$ equal to the maximum of absolute values of vector coordinates.
More specifically, let us set
$I_{D}=\{I_u=[u_0,v_0]:u\in D\}.$
Endpoints of segments from $I_{D}$ are constructed in sequential manner in polynomial time and space
by defining a new segment $I_u$ to provide general position for the set of endpoints
of the set $I_{D'}\cup\{I_u\},\,D'\subset D,$ where segments of $I_{D'}$ are already defined.
Here endpoints of $I_u$ are chosen in the rational grid that contains $u$
whose elementary cell size is $\frac{c_1}{|D|^2}\times \frac{c_1}{|D|^2}$ for some small absolute rational constant $c_1=c_1(\delta).$
Assuming $u=(u_x,u_y),$ the point $u_0$ is chosen in the lower part of the grid with $y$-coordinates less than $u_y-\delta/4$
whereas $v_0$ is taken from the upper one for which $y$-coordinates exceed $u_y+\delta/4.$

Let $S$ be the set of endpoints of segments from $I_D.$
Every disk having $I_u$ as its diameter does not contain any points
of $S$ distinct from endpoints of $I_u.$
Let $G=(S,E)$ be a Delaunay triangulation for $S$ which can be computed in polynomial time and space in $|D|.$
Obviously, each segment $I_u$
coincides with some edge from $E.$
We have $d_{\min}\leq r$ and $\mu=O(|S|).$
It remains to prove that $r\leq d_{\max}.$ Due to the remark \ref{kobylkin474757428}
and a construction of the set $D$ (see fig. 1 and proof of the theorem 1 from \cite{MasuyamaIbarakiEtAl1981})
the set $S$ can be constructed such that the inequality $r\leq d_{\max}$ holds true for $G.$
Moreover, representation length for vertices of $S$ is polynomial with respect to
representation length for points of $D.$

Let $k$ be a positive integer. Obviously, centers of at most $k$ disks of radius $r_0,$ containing $D$ in their union,
give centers of radius $r>r_0$ disks whose union is intersected with each segment from $E.$
Conversely, let $T$ be a disk of radius $r$ which intersects a subset $I_{D'}=\{I_u:u\in D'\}$ of segments
for some $D'\subseteq D.$ When $|D'|=1,$ it is easy to transform $T$
into a disk which contains the segment $I_{D'}.$
Points of $D$ have integer coordinates. Moreover, squared Euclidean distance between each pair
of points of the subset $D'$ does not exceed
$(2r_0+4\delta)^2=4r_0^2+16r_0\delta+16\delta^2.$ As $r_0\in\mathbb{Z},$
points from $D'$ are located within the distance $2r_0$ from each other.
Let us use Helly theorem. Let $R$ be the minimum radius of the disk $T_0,$ containing
any triple $u_1,\,u_2$ and $u_3$ from $D'.$
W.l.o.g. we suppose that, say, $u_1$
and $u_2$ are on the boundary of $T_0$ and denote its center by $O.$ Obviously, $R\leq r_0+2\delta.$
Let us show that the case $R>r_0$ is void.
The center of $T_0$ can be shifted along the midperpendicular to $[u_1,u_2]$
to have $u_1$ and $u_2$ at the distance $r_0$ from the shifted center $O'.$
The distance from the point $u_3$ to the radius $r_0$ circle centered at $O'$ does not exceed
$$
|O-u_3|_2+|O-O'|_2-r_0\leq 2\delta +\sqrt{(r_0+2\delta)^2-\delta_1^2}-\sqrt{r_0^2-\delta_1^2}=
$$
$$
=2\delta+\frac{4r_0\delta+4\delta^2}{\sqrt{(r_0+2\delta)^2-\delta_1^2}+\sqrt{r_0^2-\delta_1^2}}
\leq 2\delta+2\sqrt{r_0\delta+\delta^2}<\frac{1}{480r_0^5},
$$
where $\delta_1=\frac{|u_1-u_2|_2}{2}\leq r_0.$
By the lemma \ref{kobylkinks1234321} we have $R\leq r_0.$ Thus, $D'$ is contained
in some disk of radius $r_0.$ Given a set of points on the plane, the smallest radius disk
can be found in polynomial time and space which covers this set. Therefore we can convert
any set of at most $k$ disks of radius $r$ whose union is intersected with each segment from $E$
to some set of at most $k$ disks of radius $r_0$ whose union covers $D.$
\end{proof}

Using the corollary 1 of section 4.2 from \cite{journals/comgeo/ArkinAMM15} and the theorem 1 from \cite{Onoyama1984}
we arrive at the lower bound $\mu=\Omega\left(\sqrt[3]{n^2}\right)$ which holds true with positive probability
for Delaunay triangulations produced by $n$ random uniform points on the unit disk.
Thus, the order of the parameter $\mu$ for the considered class of hard instances of the {\sc IPGD}
problem is comparable with the one for random Delaunay triangulations.

\subsection{NP-hardness of {\sc IPGD} for other classes of proximity graphs}

The same proof technique could be applied for proving NP-hardness of the problem
within the other classes of proximity graphs. Let us start with some definitions.
The following graphs are connected subgraphs of Delaunay triangulations.
A plane graph $G=(S,E)$ is called a {\it Gabriel graph}
when $[u,v]\in E$ iff the disk having $[u,v]$ as its diameter
does not contain any other points of $S$ distinct from $u$ and $v.$
A {\it relative neighbourhood graph} is the plane graph $G$ with the same vertex set
for which $[u,v]\in E$ iff there is no any other point $w\in S$ such that
$w\neq u,v$ with $\max\{|u-w|_2,|v-w|_2\}<|u-v|_2.$
Finally, a plane graph is called a {\it minimum Euclidean spanning tree}
if it is the minimum weight spanning tree of the weighted complete graph $K_{|S|}$
whose vertices are points of $S$ such that
its edge weight is given by Euclidean distance between the edge endpoints.

\begin{corollary}\label{kobs2578784473828}
Both {\sc IPGD} and {\sc VRIPGD}$(\delta)$ problems are strongly NP-hard for
$r\in [d_{\min}, d_{\max}],\,\mu=O(n)$ and $\delta=\Theta(r)$
within classes of Gabriel, relative neighbourhood graphs and minimum Euclidean spanning trees as well as
for $r\in [d_{\max},\eta d_{\max}]$ and $\mu\leq 4$ within the class of nearest neighbour graphs where $\eta$ is a large constant.
\end{corollary}

\section{Positive results}

\subsection{Polynomial solvability of the {\sc IPGD} problem for large $r$}

Before presenting polynomially solvable case
of the {\sc IPGD} problem we are to take some preprocessing. It is aimed at reducing
the set of points, among which centers of radius $r$ disks are chosen, to a finite set
whose cardinality is bounded from above by some polynomial in $|E|.$

\subsubsection{Problem preprocessing.}
As was mentioned in the Introduction, the {\sc IPGD}
problem coincides with the {\sc Hitting Set} problem considered for Euclidean
$r$-neighbourhoods of graph edges which form the system denoted by ${\cal{N}}_r(E).$ Their boundaries are composed of four parts: two half-circles
and two parallel straight line segments. W.l.o.g. we can assume that intersection of any subset
of objects from ${\cal{N}}_r(E)$ (if nonempty) contains a point from the intersection of boundaries of two objects from
${\cal{N}}_r(E).$ Thus, our choice of points to form a feasible solution to the {\sc IPGD} problem
can be restricted to the set of intersection points of boundaries of pairs of objects from ${\cal{N}}_r(E).$
The following lemma can be considered a folklore.

\begin{lemma}\label{koba123}
Let $G=(V,E)$ be a simple plane graph.
Each feasible solution $C$ to the {\sc IPGD} problem for $G$ can be
converted in polynomial time and space (in $|E|)$ to a feasible solution $D\subset D_r(G)$ to
{\sc IPGD} for $G$ with $|D|\leq |C|,$ where $D_r(G)\subset\mathbb{R}^2$ is some set of cardinality of the order $O(|E|^2)$
which can be constructed in polynomial time and space.
\end{lemma}

\subsubsection{Polynomially solvable case of {\sc IPGD}.}
In distinction
to the cases where either $r\in [d_{\min},d_{\max}]$ or $r=\Theta(d_{\max})$
the {\sc IPGD} problem is polynomially solvable for $r=\Omega(R(E))$
where $R(E)$ is the smallest radius of the disk that intersects all segments from $E.$
Due to \cite{journals/computing/BhattacharyaJMR94} the {\sc IPGD} problem is solvable in $O(|E|)$ time
within the class of plane graphs for which the inequality $r\geq R(E)$ holds uniformly.

Let us consider the {\sc IPGD} problem within the class of plane graphs for which the inequality $r\geq\eta R(E)$ 
holds uniformly for some fixed constant $0<\eta<1.$
Since every radius $r$ disk contains an axis-parallel rectangle whose side is equal to $r\sqrt{2},$
roughly at most $\left\lceil\frac{\sqrt{2}R(E)}{r}\right\rceil^2\leq\left\lceil\frac{\sqrt{2}}{\eta}\right\rceil^2=k(\eta)=k$
radius $r$ disks are needed to intersect all segments from $E.$ Therefore the brute-force
search algorithm could be applied that just sequentially tries
each subset of $D_r(G)$ of cardinality at most $k.$ This amounts roughly to
$O\left(k^2|E|^{2k+1}\right)$ time complexity. Thus, we arrive at the polynomial time
algorithm whose complexity depends exponentially on $1/\eta.$
This algorithm gives an optimal solution to the {\sc IPGD} problem taking the lemma \ref{koba123} into account.

\subsection{Approximation algorithm for the {\sc IPGD} problem}

Below the approximation algorithm is reported for the {\sc IPGD} problem
whose approximation factor depends on the maximum aspect ratio among objects
of ${\cal{N}}_r(E).$ More specifically, let us focus on the case
of {\sc IPGD} where the inequality
$r\geq\frac{d_{\max}}{2\lambda}$ holds uniformly within some class ${\cal{G}}_{\lambda}$ of simple plane graphs
for a constant $\lambda>0.$
It corresponds to the
situation where objects from the system ${\cal{N}}_r(E)$ have their aspect ratio bounded from above by $1+\lambda.$
In this case it turns out that the problem admits an $O(1)$-approximation algorithm whose
factor depends on $\lambda.$ The following auxiliary problem is considered
to formulate it.

\noindent {\sc Cover endpoints of segments with disks (CESD).} Let $S(G)\subseteq V$
be the set of endpoints of edges of $G.$ It is required to find the smallest cardinality set
of radius $r$ disks whose union contains $S(G).$

\noindent {\sc Algorithm.} Compute and output $8$-approximate solution to the {\sc CESD} problem
using $O(|E|\log\,{OPT_{CESD}(S(G),r)})$-time algorithm (see sections 2 and 4 from \cite{journals/ipl/Gonzalez91}).

We call a subset $V'\subseteq V$ by a {\it vertex cover} for $G=(V,E)$
when $e\cap V'\neq\varnothing$ for any $e\in E.$
The statement below bounds the ratio of optima for {\sc CESD} and
{\sc IPGD} problems in the general case where $S(G)$ is an arbitrary vertex cover
of the graph $G.$

\begin{statement}
The following bound holds true for any graph $G\in{\cal{G}}_{\lambda}$
without isolated vertices:
$$\frac{OPT_{CESD}(S(G),r)}{OPT_{IPGD}(G,r)}\leq p(1+2\lambda)$$
where $p(x)$ is
the smallest number of unit disks needed to cover radius $x$ disk.
\end{statement}
\begin{proof}
Let $C_0=C_0(G,r)\subset\mathbb{R}^2$ be an optimal solution to {\sc IPGD}
for a given $G\in{\cal{G}}_{\lambda}.$
Set $E(c,G):=\{e\in E \colon c\in N_r(e)\},\,c\in C_0.$
For every $e\in E(c,G)$ there is a point $c(e)\in e$ with $|c-c(e)|_2\leq r.$
Any point from the set $S(c,G)$ of endpoints of segments from $E(c,G)$
is within the distance $r+d_{\max}$ from the point
$c.$
Due to definition of $p,$ at most $p(1+2\lambda)$ radius $r$ disks are needed
to cover содержит radius $r+d_{\max}$ disk.
Therefore the set $S(G)\subseteq \bigcup\limits_{c\in C_0}S(c,G)$ is contained in the union of at most $|C_0|p(1+2\lambda)$ radius $r$
disks.
\end{proof}

\begin{corollary}
The algorithm is $8p(1+2\lambda)$-approximate.
\end{corollary}

\begin{remark}
Approximation factor of the algorithm is in fact lower when ${\cal{G}}_{\lambda}$ is the
subclass of Delaunay triangulations or of their subgraphs. Indeed, in this case there is no need
to cover the whole radius $r+d_{\max}$ disk with radius $r$ disks.
\end{remark}

\begin{remark}
If $S(G)$ is the set of midpoints of segments from $E,$ the algorithm
is $8p(1+\lambda)$-approximate.
\end{remark}

\section{Conclusion}

Complexity and approximability are studied for the problem of intersecting
a structured set of straight line segments with the smallest number of
disks of radii $r>0$ where a structural information about segments
is given in the form of an edge set of a plane graph.
It is shown that the problem is strongly NP-hard within the class of Delaunay triangulations
and some of their subgraphs for small and medium values of $r$ while for large $r$
it is polynomially solvable. Fast approximation algorithm is given for the {\sc IPGD} problem
whose approximation factor depends on the maximum aspect ratio among objects from ${\cal{N}}_r(E).$
Of course,
those algorithms are of particular interest whose factor is bounded from above
by some absolute constant. This sort
of algorithms is our special
focus for future research.

\appendix

\section{Proof of the Lemma \ref{kobylkinks1234321}}

\begin{proof}
Let $u=(x,y),\,v=(x_1,y_1)$ and $w=(x_2,y_2)$ be distinct points of $X.$
Consider an arbitrary radius $r$ circle (out of two circles) which passes through
$v$ and $w,$ and denote its center by $O.$
A lower bound is obtained below for the distance
$\pi=\pi(u;v,w)$ from that circle to the point $u\notin C(v,w).$

Let $\Delta=|v-w|_2,$
$\lambda=\sqrt{r^2-\frac{  \Delta^2}{  4}}$,
$a=(u-v,u-w)$ and $b=(u-v,(v-w)^{\perp})$, where $(v-w)^{\perp}=\pm (y_1-y_2,-x_1+x_2)$.
The distance $\pi>0$ can be written in the form:
$$
\pi=\pi(u;v,w)=\left|\left|\frac{v+w}{2}-\lambda\frac{(v-w)^{\perp}}{|v-w|_2}-u\right|_2-r\right|
=\left|\frac{a+\frac{  {2\lambda b}}{  \Delta}}
 {\sqrt{a+\frac{  {2\lambda b}}{  \Delta}+r^2}+r}\right|.
$$
Without loss of generality it assumed that
$u$ is in the $2r$ radius disk centered at $O$.
Indeed, otherwise $\pi\geq r\geq\frac{  1}{  r}$.
Let us bound denominator of fraction $\pi,$
taking into account that $\Delta\leq 2r,$ $|u-v|_2\leq |u-O|_2+|O-v|_2\leq 3r$ and $|b|/\Delta\leq 3r:$
$$\sqrt{a+\frac{2\lambda b}{\Delta}+r^2}+r\leq 5r.$$
As points of $X$ have integer coordinates, $a$ and $b$ are integers.
For $\Delta^2=4r^2$ we get $\pi\geq\frac{  1}{  5r}$.
When $\Delta^2\leq 4r^2-1$ it is enough to prove the inequality
$$\left|a+\frac{2\lambda b}{\Delta}\right|\geq\frac{1}{96r^4}.$$
Indeed, again, combining this bound with the aforementioned upper bound for denominator
of the fraction $\pi,$ we get $\pi\geq\frac{1}{480r^5}$.

For integer $\frac{2\lambda b}{\Delta}$ the left-hand side of the inequality
is at least 1. Thus, it remains for us to prove the inequality for the case where 
$\frac{2\lambda b}{\Delta}\notin\mathbb{Z}$.
Suppose that $q=\left\{\left|\frac{2\lambda b}{\Delta}\right|\right\}>0$ and
$k=\left[\left|\frac{2\lambda b}{\Delta}\right|\right],$ where $\{\cdot\}$ and
$[\cdot]$ denote fractional and integer part of real number respectively.
In fact, the term $\min\{q,1-q\}$ is bounded from below.
Let us start estimating with $q.$
First, it is assumed that $\gamma=\frac{  4r^2b^2}{  \Delta^2}\in~\mathbb{Z}.$
We have $k^2<\frac{  4\lambda^2b^2}{  \Delta^2}<(k+1)^2.$
As $q>0,$ we get $q\geq\left\{\sqrt{k^2+1}\right\}.$
Due to concavity of the square root we have
$$
\left\{\sqrt{k^2+1}\right\}
=\left\{\sqrt{\frac{2k\cdot k^2}{2k+1}+
\frac{(k+1)^2}{2k+1}}\right\}
\geq\left\{k+\frac{1}{2k+1}\right\}
=\frac{1}{2k+1}
\geq$$$$\geq\frac{1}{\frac{  4\lambda |b|}{  \Delta}+1}
\geq\frac{  1}{  13r^2}.
$$

Now the case is considered where
$\gamma\notin\mathbb{Z}$.
As $2kq+q^2\geq \{2kq+q^2\}=\{\gamma\}$, we have that
$$
q\geq \sqrt{k^2+\{\gamma\}}-k\geq\frac{\{\gamma\}}
{\sqrt{k^2+\{\gamma\}+k}}
\geq\frac{\frac{  1}{  \Delta^2}}{\frac{  4r|b|}{  \Delta}}
\geq \frac{1}{12r^2\Delta^2}\geq \frac{1}{48r^4}.
$$

Let us get a lower bound for $1-q.$
Again, assume that $\gamma\in\mathbb{Z}$.
Arguing analogously, we arrive at the bound
$$
2k(1-q)+(1-q)^2\geq\{(k+1-q)^2\}
=\left\{(k+1)^2-\frac{4\lambda^2b^2}{\Delta^2}-2q(1-q)\right\}\geq\frac{1}{2}.
$$
Resolving the quadratic inequality with respect to
$1-q$, we get:
$$
1-q\geq\sqrt{k^2+\frac{1}{2}}-k
=\frac{\frac{1}{2}}{\sqrt{k^2+\frac{1}{2}}+k}
\geq\frac{1}{\frac{  8r|b|}{  \Delta}}\geq\frac{1}{24r^2}.$$
Now let $\gamma\notin\mathbb{Z}$.
Let us consider the subcase where $\{\gamma\}+2q(1-q)>1$.
We get
$$
\left\{(k+1)^2-\frac{4\lambda^2b^2}{\Delta^2}-2q(1-q)\right\}\geq 1-\{\gamma\}\geq\frac{1}{\Delta^2}.
$$
Resolving the corresponding inequality with respect to $1-q$,
we arrive at the analogous lower bound $1-q\geq\frac{1}{48r^4}$.

Now we are to address the case where $\{\gamma\}+2q(1-q)<1$. Obviously
$$
\left\{(k+1)^2-\frac{4\lambda^2b^2}{\Delta^2}-2q(1-q)\right\}
=1-\{\gamma\}-2q(1-q).
$$
For $1-q<\frac{  1}{  4\Delta^2}$ we have
$1-\{\gamma\}-2q(1-q)
\geq\frac{  1}{  \Delta^2}-\frac{  1}{  2\Delta^2}
=\frac{  1}{  2\Delta^2}.$
Arguing analogously, we obtain the following bound $1-q\geq\frac{  1}{  96r^4}$;
otherwise
$1-q\geq\frac{  1}{  4\Delta^2}\geq\frac{  1}{  16r^2}$.

For $\{\gamma\}+2q(1-q)=1$ we have:
$$
1-q=\frac{1-\{\gamma\}}{  2q}
\geq \frac{1-\{\gamma\}}{2}\geq \frac{1}{2\Delta^2}\geq\frac{1}{8r^2}.
$$
Finally, we get the claimed bound $\min\{q,1-q\}\geq\frac{  1}{  96r^4}$.
\end{proof}

\end{document}